\documentclass[12pt]{article}
\usepackage{fullpage}
\usepackage{here}

\usepackage{amsthm,amsmath,amssymb}
\usepackage{xcolor}
\usepackage{graphicx}
\usepackage{booktabs}
\usepackage{caption}
\usepackage{hyperref}
\usepackage{enumerate}
\usepackage{enumitem}
\usepackage{lineno}
\usepackage[noend]{algpseudocode}
\makeatletter
\def\BState{\State\hskip-\ALG@thistlm}
\makeatother

\newcommand{\REMOVE}[1]{}

\newcommand{\eps}{\varepsilon}
\renewcommand{\phi}{\varphi}

\newcommand{\dist}{\textsf{dist}}
\newcommand{\diam}{\textsf{diam}}
\newcommand{\proj}{\textsf{proj}}
\newcommand{\per}{\textsf{per}}
\newcommand{\opt}{\textsf{OPT}}
\newcommand{\R}{\mathbb{R}}
\newcommand{\cut}{\delta_{\operatorname{cut}}}
\newcommand{\hausdorff}{\delta_{\operatorname{H}}}
\newcommand{\rock}{\delta_{\operatorname{rock}}}
\newcommand{\station}{\delta_{\operatorname{station}}}
\newcommand{\frechet}{\delta_{\operatorname{F}}}

\renewcommand{\emph}{\textbf}

\theoremstyle{plain}
\newtheorem{theorem}{Theorem}
\newtheorem{lemma}[theorem]{Lemma}

\newtheorem{definition}[theorem]{Definition}
\newtheorem{obs}{Observation}

\newcommand{\ShoLong}[2]{#2}

\title{Rock Climber Distance: Frogs versus Dogs\thanks{Research supported in part by the NSF awards CCF-1422311 and CCF-1423615.}}
\author{
Hugo A. Akitaya\thanks{Department of Computer Science, Tufts University, Medford, MA, USA.}
\and
Leonie Ryvkin\thanks{Department of Mathematics, Ruhr University Bochum, Germany}
\and
Csaba D. T\'oth\footnotemark[2]~\thanks{Department of Mathematics, California State University Northridge, Los Angeles, CA, USA.}
}
\date{}

\begin{document}
\maketitle

\begin{abstract}
The classical measure of similarity between two polygonal chains in Euclidean space is the Fr\'echet distance, which corresponds to the coordinated motion of two mobile agents along the chains while minimizing their maximum distance. As computing the Fr\'echet distance takes near-quadratic time under the Strong Exponential Time Hypothesis (SETH), we explore two new distance measures, called \emph{rock climber distance} and \emph{$k$-station distance}, in which the agents move alternately in their coordinated motion that traverses the polygonal chains. We show that the new variants are equivalent to the Fr\'echet or the Hausdorff distance if the number of moves is unlimited. When the number of moves is limited to a given parameter $k$, we show that it is NP-hard to determine the distance between two curves. We also describe a 2-approximation algorithm to find the minimum $k$ for which the distance drops below a given threshold.
\end{abstract}

\section{Introduction}

Recognizing similarity between geometric objects is a classical problem in pattern matching, and has recently gained renewed attention due to its applications in  artificial intelligence and robotics. Statistical methods and the Hausdorff distance have proved to be good similarity measures for static objects, but are insensitive to spatio-temporal data, such as individual trajectories or clusters (flocks) of trajectories. The \emph{Fr\'echet distance} (defined below) is considered to be one of the best similarity measures between curves in space. Between two polygonal chains with a total of $n$ vertices, the Fr\'echet distance can be computed in $O(n^2\ {\rm polylog}\ n)$ time~\cite{altgodau,Buchin2017}. Under the Strong Exponential Time Hypothesis (SETH), there is a lower bound of $\Omega(n^{2-\delta})$, for any $\delta>0$, for computing the Fr\'echet distance~\cite{Bringmann}, or even approximating it within a factor of 3~\cite{DBLP:conf/soda/BuchinOS19}. Without SETH, the current best lower bound for the time complexity under the algebraic decision tree model is $\Omega(n\log n)$ \cite{walkdog}.

Applications, however, call for efficient algorithms for massive trajectory data. This motivates the quest for new variants of the Fr\'echet distance that may bypass some of its computational bottlenecks but maintain approximation guarantees.

In this paper, we introduce the \emph{rock climber} distance. It combines
properties of the continuous and the discrete Fr\'echet distance, and is closely related to the recently introduced $k$-Fr\'echet distance~\cite{DBLP:journals/corr/abs-1903-02353}. 
The classic \emph{Fr\'echet distance} corresponds to coordinated motion, where two agents follow the polygonal paths $P$ and $Q$, so that they minimize the maximum distance between the agents (intuitively, the agents are a man and a dog, and they minimize the length of the leash between them). The \emph{discrete Fr\'echet distance} considers discrete motion on the vertices of the two chains (i.e., walking a frog~\cite{SarielBook}, pun intended). The \emph{rock climber distance} corresponds to a coordinated motion of two agents along $P$ and $Q$ that is continuous, but only one agent moves at a time, hence it can be described by an axis-parallel path in a suitable parameter space (the so-called free space diagram, described below). 

\paragraph{Definitions}
Given two polygonal chains, parameterized by piecewise linear curves, 
$P: [0,1] \to \R^2$ and $Q: [0,1] \to \R^2$, 
the Hausdorff distance is defined as  
\ShoLong{

\vspace{-\baselineskip}
\begin{align*}
    \hausdorff(P,Q) =\max &\{ 
\max_{s\in [0,1]} \min_{t\in [0,1]} \|P(s)-Q(t)\|, \\
& \max_{t\in [0,1]} \min_{s\in [0,1]} \|P(s)-Q(t)\|
\}.
\end{align*}
}{
\[
\hausdorff(P,Q)
=\max\{ 
\max_{s\in [0,1]} \min_{t\in [0,1]} \|P(s)-Q(t)\|,
\max_{t\in [0,1]} \min_{s\in [0,1]} \|P(s)-Q(t)\|
\}.
\]
}
and the Fr\'echet distance is defined as
\[
\frechet(P,Q)=\inf_{\sigma,\tau} \max_{t \in [0,1]} \lVert P(\sigma(t))-Q(\tau(t)) \rVert,
\]

where $\sigma,\tau:[0,1] \to [0,1]$ range over all orientation-preserving homeomorphisms of $[0,1]$. 
The standard machinery for finding nearby points in the two polygonal chains, introduced by Alt and Godau~\cite{altgodau} uses the so-called \emph{free space diagram}. For every $\eps>0$, the \emph{free space} is defined as 
\[
F_\eps(P,Q)=\{(s,t)\in [0,1]^2 \colon \| P(s)-Q(t) \| \leq \eps\}.
\]
Note that $F_\eps(P,Q)\subset [0,1]^2$, where a point $(s,t)\in [0,1]^2$ corresponds to the positions $P(s)$ and $Q(t)$ on the two chains. 
The Fr\'echet distance between $P$ and $Q$ is at most $\eps$ if and only if
the free space contains a strictly $x$- and $y$-monotone path from $(0,0)$ to $(1,1)$;
namely, $\gamma:[0,1]\rightarrow [0,1]^2$, $\gamma(t)=(\sigma(t),\tau(t))$.

We define further terms connected to the free space diagram below: A \emph{component} of a free space diagram is a connected subset $c\subseteq F_\varepsilon(P,Q)$.
A set $S$ of components \emph{covers} a set $I\subseteq [0,1]_P$ of the parameter space (corresponding to the curve $P$) if $I$ is a subset of the projection of $S$ onto said parameter space, i.e., $\forall x\in I \colon \exists c\in S, y\in [0,1]_Q \colon (x,y)\in c$.
Covering on the second parameter space is defined analogously.

\paragraph{Rock Climbers Distance.} 
Assume that two rock climbers each choose a route on a vertical wall, represented by polygonal chains $P$ and $Q$. They secure each other with a rope: While one endpoint of the rope is  firmly attached to the rock, the other endpoint may move. Both climbers must be secured at all times, and so only one climber can move at a time. The \emph{rock climber distance} is the minimum length of a rope that allows them to traverse the routes $P$ and $Q$, that is, 
\begin{equation}
\rock(P,Q)=\inf_{\gamma} \max_{t \in [0,1]} \lVert P(\sigma(t))-Q(\tau(t)) \rVert,\label{eq:rock}
\end{equation}
where $\gamma:[0,1]\rightarrow [0,1]^2$, $\gamma(t)=(\sigma(t),\tau(t))$, ranges over all $x$- and $y$-monotonically increasing axis-parallel paths from $(0,0)$ to $(1,1)$. 

We show that $\rock(P,Q)=\frechet(P,Q)$ (cf.~Theorem~\ref{thm:feasible1}), albeit the number of turns of the path $\gamma$ may far exceed the number of vertices of $P$ and $Q$. 
This indicates that the number of axis-parallel segments in $\gamma$ is a crucial parameter.
For every $k\in \mathbb{N}$, we define $\rock(k,P,Q)$ by equation~\eqref{eq:rock} with the additional condition that the path $\gamma$ consists of at most $k$ line segments. 

\paragraph{Rock Climber Distance with $k$ Stations.}
The main focus of this paper is a variant of the rock climber distance, where 
the number of axis-parallel segments is a fixed parameter $k$, but these segments need not form a continuous path from $(0,0)$ to $(1,1)$.
Assume that a rock climber club decides to install permanent safety ropes along the routes $P$ and $Q$ for training purposes. Each rope has one fixed endpoint on $P$ or $Q$, and its other endpoint can move freely on some subcurve of the other polygonal chain ($Q$ or $P$, respectively). The mobile endpoint of a rope, however, cannot pass through the fixed endpoint of another rope. The club decides to install $k\in \mathbb{N}$ identical ropes: What is the minimum length of a rope that allows safe traversal on both $P$ and $Q$? More formally, we arrive at the following definition.

\begin{definition}
For two polygonal chains, $P$ and $Q$, and an integer $k\in \mathbb{N}$, 
the \emph{$k$-station distance}, denoted $\station(k,P,Q)$, 
is the infimum of all $\eps>0$ such that 
there exist two subdivisions $0=a_0<a_1<\ldots <a_p=1$ and $0=b_0<b_1<\ldots <b_q=1$
into a total of $p+q=k$ intervals such that 
\[\min_{j\in \{1,\ldots , q\}}
\min_{s\in [a_{i-1},a_i]}
\|P(s)-Q(b_j)\| \leq \eps
\mbox{ \rm for }
i=1,\ldots , p;
\]
\[\min_{i\in \{1,\ldots , p\}}
\min_{t\in [b_{j-1},b_j]}
\|P(a_i)-Q(t)\| \leq \eps
\mbox{ \rm for }
j=1,\ldots , q.
\]
\end{definition}

Every subcurve $P[a_{i-1},a_i]$ of $P$ has some closest point $Q(b_{j(i)})$ in $Q$; and 
every subcurve $Q[b_{j-1},b_j]$ of $Q$ has a closest point $P(b_{i(j)})$ in $P$. 
In the free~space diagram $F_\eps(P,Q)$, where $\eps=\station(k,P,Q)$, the union of   
horizontal segments $[a_{i-1},a_i]\times \{b_{j(i)}\}$ and 
vertical segments $\{a_{i(j)}\}\times [b_{j-1},b_j]$
projects surjectively to the unit interval $[0,1]$ on each coordinate axis.

\paragraph{Fr\'echet Distance with $k$ Jumps.}
The $k$-station distance can also be considered as a variant of the \emph{$k$-Fr\'echet distance}, introduced by Buchin and Ryvkin~\cite{eurocg} (see also~\cite{DBLP:journals/corr/abs-1903-02353}). 
Intuitively, it measures the similarity between two polygonal chains after  $k$ ``mutations.'' Formally, $\cut(k,P,Q)$ is the infimum of $\eps>0$ such that $P$ and $Q$ can each be subdivided into $k$ subcurves, $P_i$ and $Q_i$ ($i=1,\ldots ,k)$,  where $\frechet(P_i,Q_{\pi(i)})\leq \eps$ for some permutation $\pi:[k]\rightarrow [k]$.
Importantly, the chains $P$ and $Q$ can be subdivided at any point, not only at vertices. 
Determining the minimum $k\in \mathbb{N}$ for which $\cut(k,P,Q)\leq \eps$
for a given $\eps$ is NP-hard, and conjectured to be $\exists \mathbb{R}$-hard. 
%
The $k$-station distance can be considered as a restricted version of the $k$-Fr\'echet distance, where either $P_i$ or $Q_{\pi(i)}$ is required to be a single point
(i.e., a trivial curve) for $i=1,\ldots ,k$. By definition, we have 
$\cut(k,P,Q)\leq \station(k,P,Q)$ for all $k\in \mathbb{N}$.

\paragraph{Unit Disk Cover (UDC).}
The rock climber $k$-station distance is also reminiscent of the unit disk cover problem: Given a point set $S\subset \mathbb{R}^2$, find a minimum set $\mathcal{D}$ of unit disks
such that $S\subset \bigcup \mathcal{D}$. When $S$ is finite, UDC is known to be NP-hard~\cite{DBLP:journals/ipl/FowlerPT81}, one can find a 4-approximation in $O(n\log n)$ time~\cite{DBLP:journals/comgeo/BiniazLMS17}. In the \emph{Discrete Unit Disk Cover} problem, $S$ is finite, and the disks are restricted to a finite set of possible centers~\cite{DBLP:journals/jda/BasappaAD15}; the discretized version admits a PTAS via local search~\cite{DBLP:journals/dcg/MustafaR10,DBLP:journals/dcg/RoyGRR18}. These results extend to the cases where $S$ is confined to a narrow strip~\cite{DBLP:journals/tcs/FraserL17}, or $S$ is a finite union of line segments~\cite{DBLP:conf/caldam/Basappa18}.
Finding the minimum $k\in \mathbb{N}$ such that $\station(k,P,Q)\leq 1$
can be considered as a variant of UDC, where $P$ (resp., $Q$) must be covered by disks centered at $Q$ (resp., $P$), and each disk can cover at most one contiguous arc of a curve. 

\paragraph{Our Results.} 
\ShoLong{We}{In this paper, we} 
prove the following results.
\begin{enumerate}
    \item We show that $\rock(P,Q)=\frechet(P,Q)$ and $\station(P,Q)=\hausdorff(P,Q)$
    for a sufficiently large $k$ (that depends on $P$ and $Q$). It follows that for any two polygonal chains, $P$ and $Q$, there exists a positive integer $k$ such that $\cut(k,P,Q)\leq \frechet(P,Q)$. The first identity implies that $\rock(P,Q)$ can be computed in $O(n^2\sqrt{\log n}(\log\log)^{3/2})$ time~\cite{Buchin2017},
    where $P$ and $Q$ jointly have $n$ vertices
    (Section~\ref{sec:feasibility}).
    \item We prove that it is NP-complete to decide whether $\station(k,P,Q)\leq \eps$ for two given polygonal chains, $P$ and $Q$, and parameters $k$ and $\eps>0$ (Section~\ref{sec:hard}). 
    \item We also give a 2-approximation algorithm for finding the minimum $k\in \mathbb{N}$ such that $\station(P,Q,k)\leq \eps$ for given polygonal chains $P$ and $Q$, and a threshold $\eps>0$. We reduce the problem to a variant of the set cover problem over axis-parallel line segments, for which a greedy strategy yields a 2-approximation (Section~\ref{sec:apx}). 
\end{enumerate}

\paragraph{Further Related Previous Work.}
Alt, Knauer, and Wenk~\cite{knauer} compared the Hausdorff to the Fr\'echet distance and discussed $\kappa$-bounded curves as a special input instance. In particular, they showed that for convex closed curves Hausdorff distance equals Fr\'echet distance. 
For curves in one dimension Buchin et al.~\cite{walkdog} proved equality of Hausdorff and weak Fr\'echet distance using the well-known Mountain Climbing theorem~\cite{mountain}. 
Recently, Driemel et al.~\cite{DBLP:journals/corr/abs-1903-03211} gave bounds on the VC-dimension of curves under Hausdorff and Fr\'echet distances. 
Buchin~\cite{b-cfdts-07} characterized these measures in terms of the free space, which motivated the study of the variants of the $k$-Fr\'echet distance; see also Har-Peled and Raichel~\cite{DBLP:journals/talg/Har-PeledR14} for a treatment using product spaces.
The $k$-station distance is also related to partial curve matching, studied by Buchin, Buchin, and Wang~\cite{partialcurve}, who  presented a polynomial-time algorithm to compute the ``partial Fr\'echet similarity.'' A variation of this similarity was considered by Scheffer~\cite{scheffer}.


\section{Relations to Other Distance Measures}
\label{sec:feasibility}

In this section, we compare the rock climber distance and the $k$-station distance to the Fr\'echet and Hausdorff distances, as well as the cut distance. 

\paragraph{Preliminaries.}
Let $P:[0,1]\rightarrow \mathbb{R}^2$ and $Q:[0,1]\rightarrow \mathbb{R}^2$ two piecewiese linear curves. That is, there are subdivisions $0=a_0<a_1<\ldots <a_m=1$ and $0=b_0<b_1<\ldots <b_n=1$ such that $P, Q$ are linear on each subinterval $[a_{i-1},a_i]$ and $[b_{j-1},b_j]$, respectively.  Recall that for every $\eps>0$, the free space is defined as 
$F_\eps(P,Q)=\{(s,t)\in [0,1]^2 \colon \| P(s)-Q(t) \| \leq \eps\}$,
which is a subset of the configuration space $U=[0,1]^2$. We 
can subdivide $U$ into $mn$ \emph{cells} of the form $C_{i,j}=[a_{i-1},a_i]\times [b_{j-1},b_j]$, for $i=1,\ldots, m$ and $j=1,\ldots , n$. 
It is known that $C_{i,j}\cap F_\eps(P,Q)= C_{i,j}\cap E_{i,j}$, where 
$E_{i,j}$ is either an ellipse or a slab parallel to the diagonal of $C_{i,j}$
(in case $P([a_{i-1},a_i])$ and $Q([b_{j-1},b_k])$
are parallel line segments).

\paragraph{Geometric Properties.} We prove a few elementary properties for monotone curves passing through a cell of the free space diagram. We start with an easy observation.
\ShoLong{Omitted proofs are available in the Appendix.}{}

\begin{lemma}\label{lem:ellpise}
Let $E$ be an ellipse with maximal curvature $\kappa$.
Then for every point $p\in \partial E$, there are horizontal and vertical segments 
$H_p$ and $V_p$, respectively, such that $p\in H_p\subset E$, $p\in V_p\subset E$,
and $\|H_p\|+\|V_p\|\geq 2/\kappa$.
\end{lemma}

\global\def \proofLemEllipse {
For every point $p\in \partial E$, there is a disk $D_p$ of radius $\frac{1}{\kappa}$ such that $p\in D_p\subset E$. Let $H_p$ and $V_p$, respectively, be the maximal horizontal and vertical segments that lie in $D_p$ and contain $p$. Since $H_p$ and $V_p$ are orthogonal, they form a right triangle with hypotenuse $\diam(D_p)=2/\kappa$. The triangle inequality yields $\|H_p\|+\|V_p\|\geq 2/\kappa$.
}

\ShoLong{}{
\begin{proof}
\proofLemEllipse
\end{proof}
}

\begin{lemma}\label{lem:monotone}
Let $C$ be an axis-aligned rectangle and $E$ an ellipse such that $C\cap E\neq\emptyset$.
Let $\alpha:[0,1]\rightarrow C\cap E$ be an $x$- and $y$-monotone increasing curve.
Then there exists an $x$- and $y$-monotone increasing curve $\beta:[0,1]\rightarrow C\cap E$ such that $\beta(0)=\alpha(0)$, $\beta(1)=\alpha(1)$, and (the image of) $\beta$ is 
a polygonal chain consisting of a finite number of axis-parallel edges.
\end{lemma}

\global\def \proofLemMonotone {
Note that every axis-parallel line passing through the interior of $C$ subdivides $C$ into two axis-aligned rectangles; and every axis-parallel line passing through an interior point of $\alpha$ subdivides $\alpha$ into two $x$- and $y$-monotone curves. It is enough to prove the claim in each cell of a finite arrangement of axis-parallel lines.

The axis-parallel lines passing through the four extreme points of $E$ (i.e., the leftmost, rightmost, lowest, and highest points) subdivide $\partial E$ into $x$- and $y$-monotone arcs. Assume without loss of generality that these lines do not intersect the interior of $C$. Further assume, by subdividing along the axis-parallel lines passing through the endpoints of $\alpha$, that $\alpha(0)$ and $\alpha(1)$, respectively, are the lower-left and upper-right corner of $C$. 
Note that both $C$ and $E$ are convex, hence $C\cap E$ is convex. If the upper-left or the lower-right corner of $C$ is in $E$, then the the two adjacent sides of $C$ are in $C\cap E$, and form an axis-parallel path with two edges from the lower-left to the upper right corner of $C$ 

Assume that neither the upper-left nor the lower-right corner of $C$ is in $E$. Construct an $x$- and $y$-monotone increasing curve $\beta:[0,1]\rightarrow C\cap E$ from the lower-left to the upper right corner of $C$ greedily as follows: Start the path from the lower-left corner $p_0$, and alternately append maximal horizontal and vertical segments in $E\cap C$ to the current endpoint until reaching the upper right corner. By Lemma~\ref{lem:ellpise}, 
the combined length of any two consecutive edges, excluding the first and last two edges, 
is at least $2/\kappa$, where $\kappa>0$ is a constant that depends only on $E$. It follows that the path reaches the upper right corner within at most $\kappa\cdot \per(C)+4$ iterations.
}

\ShoLong{}{
\begin{proof}
\proofLemMonotone
\end{proof}
}

For a set $S\subset \mathbb{R}^2$, let $\proj_x(S)$ and $\proj_y(S)$ denote the orthogonal projection of $S$ onto the $x$- and the $y$-axis, respectively. 

\begin{lemma}\label{lem:cover}
Let $C$ be an axis-aligned rectangle and $E$ an ellipse such that $C\cap E\neq\emptyset$.
Then there exists a finite set $\mathcal{S}$ of axis-parallel line segments 
in $C\cap E$ such that 
$\proj_x(C\cap E)=\proj_x(\bigcup \mathcal{S})$ and 
$\proj_y(C\cap E)=\proj_y(\bigcup \mathcal{S})$.
\end{lemma}

\global\def \proofLemCover {
\begin{proof}
Let $a$, $b$, $c$, and $d$, respectively, be a leftmost, rightmost, lowest, and highest point in $C\cap E$. By convexity, we have $ab,cd\subset C\cap E$. Note that $\proj_x(C\cap E)=\proj_x(ab)$ and $\proj_y(C\cap E)=\proj_y(cd)$. The segments $ab$ and $cd$ yield $x$- and $y$-monotone curves between their endpoints. By Lemma~\ref{lem:monotone}, $C\cap E$ contains an $ab$-path and an $cd$-path that are $x$- and $y$-monotone, and have a finite number of edges. We conclude by taking $S$ to be the union of all edges of these paths.
\end{proof}
}
\ShoLong{}{
\begin{proof}
\proofLemCover
\end{proof}
}

\paragraph{Relation to the Fr\'echet Distance.}
We show that the rock climber distance equals the Fr\'echet distance.

\begin{theorem}\label{thm:feasible1}
For two polygonal chains, $P$ and $Q$, it holds that $\rock(P,Q)=\frechet(P,Q)$.
\end{theorem}

\global\def \proofThmFeasible {
We first prove $\rock(P,Q)\leq \frechet(P,Q)$. Put $\eps:=\frechet(P,Q)$. Let $\alpha:[0,1]\rightarrow F_\eps(P,Q)$ be a strictly $x$- and $y$-monotone increasing curve
from $(0,0)$ to $(1,1)$. If $P$ and $Q$ contain segments at distance precisely $\eps$ apart, then the free space $F_\eps(P,Q)$ would contain line segments in some cells. To avoid dealing with such cells, we inflate the free space as follows. Let $D$ be the set of distances between parallel edges from $P$ and $Q$, respectively. Since $D\subset \mathbb{R}$ is finite, there exists a sufficiently small $\delta_0>0$ such that all distances in $D$ are outside of the interval $(\eps,\eps+\delta_0)$. 
Then for every $\delta\in (0,\delta_0)$, the free space $F_{\eps+\delta}(P,Q)$ is the union of regions $C_{i,j}\cap E_{i,j}$, where $C_{i,j}$ is an axis-aligned rectangle (cell), and $E_{i,j}$ is an ellipse or a parallel strip; note that $F_\eps\subset F_{\eps+\delta}$.
By Lemma~\ref{lem:monotone}, each subcurve $\gamma\cap C_{i,j}$ can be replaced 
by an $x$- and $y$-monotone polygonal chain in $C_{i,j}\cap F_{\eps+\delta}(P,Q)$ 
with the same endpoints and with a finite number of axis-parallel edges. The concatenation of these paths is an $x$- and $y$-monotone polygonal chain in $F_{\eps+\delta}(P,Q)$ from $(0,0)$ to $(1,1)$, also with a finite number of axis-parallel edges. Consequently, $\rock(P,Q)\leq \eps+\delta = \frechet(P,Q)+\delta$ for all $\delta>0$, which in turn implies  $\rock(P,Q)\leq \frechet(P,Q)$.

It remains to prove $\frechet(P,Q)\leq \rock(P,Q)$. Put $\eps:=\rock(P,Q)$. Then the free space $F_\eps(P,Q)$ contains an $x$- and $y$-monotone staircase path $\gamma$ from $(0,0)$ to $(1,1)$. For every $\delta>0$, we can perturb $\gamma$ into a strictly $x$- and $y$-monotone curve from from $(0,0)$ to $(1,1)$ in $F_{\eps+\delta}(P,Q)$. Consequently, $\frechet(P,Q)\leq \eps+\delta = \rock(P,Q)+\delta$ for every $\delta>0$, which readily implies $\frechet(P,Q)\leq \rock(P,Q)$. 
}
\ShoLong{
}{
\begin{proof}
\proofThmFeasible
\end{proof}
}

For two polygonal chains, $P$ and $Q$, with a total of $n$ segments, $\frechet(P,Q)$ can be computed in $O(n^2\sqrt{\log n}(\log\log)^{3/2})$ time~\cite{Buchin2017}. Consequently, $\rock(P,Q)$
can be computed in the same time, regardless of the complexity of the 
path $\gamma$ in $F_{\eps+\delta}(P,Q)$. 

\paragraph{Relation to the Hausdorff and $k$-Fr\'echet Distances.}
The $k$-station distance between $P$ and $Q$ equals their Hausdorff distance for a sufficiently large integer $k$.

\begin{theorem}\label{thm:hausdorff}
For two polygonal chains, $P$ and $Q$, and for $\eps>0$, there exists a $k\in \mathbb{N}$ such that $\station(k,P,Q) =\hausdorff(P,Q)$. 
\end{theorem}

\global\def \proofHausdorff {
We prove $\station(k,P,Q)\leq \hausdorff(P,Q)$ for a sufficiently large $k\in\mathbb{N}$. Put $\eps:=\hausdorff(P,Q)$. If $P$ and $Q$ contain segments at distance precisely $\eps$ apart, then the free space $F_\eps(P,Q)$ would contain line segments in some cells.  There is a $\delta_0>0$ such that the distance of any two parallel edges of $P$ and $Q$ are outside of the interval $(\eps,\eps+\delta_0)$.

Consider the free space $F_{\eps+\delta}(P,Q)$ for some $\delta\in (0,\delta_0)$. 
By Lemma~\ref{lem:cover}, there is a finite set $\mathcal{S}$ of axis-parallel segments 
whose orthogonal projections to each coordinate axis is the same as the projection of the free space $F_{\eps+\delta}(P,Q)$, that is, $\proj_x(C\cap E)=\proj_x(\bigcup \mathcal{S})$ and 
$\proj_y(C\cap E)=\proj_y(\bigcup \mathcal{S})$.  
The set $\mathcal{S}$ confirms that $\station(P,Q)\leq \eps+\delta = \hausdorff(P,Q)+\delta$ for every $\delta>0$, hence $\station(P,Q)\leq \hausdorff(P,Q)$.

Finally, we show that $\hausdorff(P,Q)\leq \station(k,P,Q)$ for all $k\in \mathbb{N}$.
Indeed, put $\hausdorff(P,Q)=\eps$. Then at least one of the curves contains a point at distance $\eps$ from the other curve. Without loss of generality, assume $p\in P$ and $\dist(p,Q)=\eps$. Regardless of the subdividion of $P$ and $Q$ into $k$ subcurves,
we have $\delta_F(P_i,Q_{\pi(i)})\geq \eps$ for the subcurve $P_i$ that contains $p$.
Consequently, $\station(k,P,Q)\geq \cut(k,P,Q)\geq \eps=\hausdorff(P,Q)$ 
for all $k\in \mathbb{N}$.
}
\ShoLong{
}
{
\begin{proof}
\proofHausdorff
\end{proof}
}


\paragraph{Remark.} 
\ShoLong{In the proofs of Theorems~\ref{thm:feasible1} and~\ref{thm:hausdorff} (see Appendix), we ``inflate'' the free space $F_\eps(P,Q)$ into $F_{\eps+\delta}(P,Q)$, $\delta>0$, to avoid the case that $P$ and $Q$ contain parallel segments at distance $\eps$.}
{In the proofs of Theorems~\ref{thm:feasible1} and~\ref{thm:hausdorff}, we have ``inflated'' the free space $F_\eps(P,Q)$ into $F_{\eps+\delta}(P,Q)$, $\delta>0$, to avoid the case that $P$ and $Q$ contain parallel segments at distance $\eps$. }
This step is necessary, as the free space $F_\eps(P,Q)$, where $\eps=\frechet(P,Q)$, need not contain an axis-parallel path from $(0,0)$ to $(1,1)$. In the simplest example, $P$ and $Q$ are two parallel segments: The free space consists only of the straight line segment at the diagonal of $[0,1]^2$.
	
	\begin{figure}[ht]
		\centering
		\includegraphics[scale=0.9]{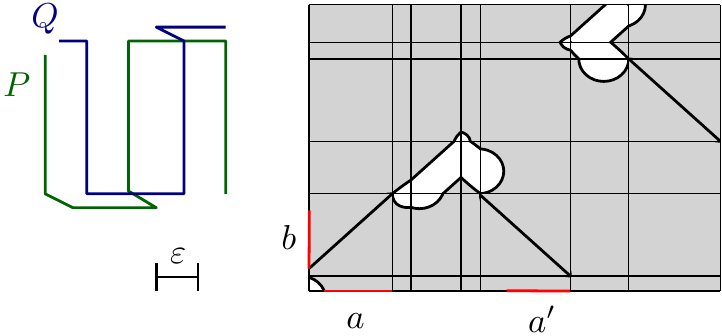}
		\caption{To project onto intervals $a$ and $a'$ we need to use the two straight line components above them, but then $b$ has two preimages for its projection.}\label{fig:feas_example}
	\end{figure}
	
Figure~\ref{fig:feas_example} shows an example where three segments in $P$ are parallel to two segments in $Q$ at distance $\eps$ apart.
It is impossible to cut $P$ and $Q$ into $k\in \{2,3\}$ pieces such that 
$\cut(k,P,Q)\approx \hausdorff(P,Q)$. However, if we allow an arbitrarily large $k\in \mathbb{N}$,
it is possible to place multiple cuts within a tiny distance in order to make sure that both parameter spaces can be covered by tiny slices of components.

\paragraph{Remark.} 
For two polygonal chains, $P$ and $Q$, with a total of $n$ segments, the free space 
$F_\eps(P,Q)$ is bounded by $N=O(n^2)$ line segments and elliptical arcs for every $\eps>0$. 
Mitchell et al.~\cite{DBLP:journals/comgeo/MitchellPS14,DBLP:journals/algorithmica/MitchellPSW19} proved that the rectilinear link distance between two points in a rectilinear polygonal domain with $N$ vertices can be computed in $O(N\log N)$ time. Perhaps this method can be adapted to decide whether the rectilinear link distance between $(0,0)$ and $(1,1)$ in the free space $F_\eps(P,Q)$ does not exceed a given parameter in time polynomial in $k$ and $n$. 
One could then find the infimum of $\eps>0$ such that $F_\eps(P,Q)$ contains such a path with $k$ or fewer links by parametric search~\cite{DBLP:journals/comgeo/OostrumV04}, and compute $\rock(k,P,Q)$ in polynomial time.

\section{NP-Hardness}
\label{sec:hard}

The $k$-station distance raises several optimization problems.
\begin{itemize}
    \item Can we find the minimum $\eps>0$ such that $\station(k,P,Q)\leq \eps$ for two polygonal chains $P$ and $Q$, and an integer $k$?
    \item Can we find the minimum $k\in \mathbb{N}$ for a given threshold $\eps>0$?
\end{itemize}
In this section, we show that the decision versions of these problems are NP-hard.
That is, it is NP-hard to decide whether $\station(k,P,Q)\leq \eps$. 
Our reduction will produce weakly simple polygonal chains $P$ and $Q$.
A polygonal chain is \emph{weakly simple} if its vertices can be moved by some arbitrary small amount to produce a Jordan arc~\cite{DBLP:journals/dcg/AkitayaAET17,DBLP:conf/soda/ChangEX15}.

We reduce from \textsc{Planar-Rectilinear-3SAT} which is NP-complete~\cite{knuth1992problem}. 
An instance of \textsc{Planar-Rectilinear-3SAT} is defined by a boolean formula $\Phi$ in 3-CNF with $n$ variables and $m$ clauses. The formula is accompanied by a planar rectilinear drawing of the bipartite graph between variables and clauses in an integer grid where all variables are represented by points on the $x$-axis, and edges do not cross this axis. 
The problem asks whether there is an assignment from the variable set to  $\{\texttt{true, false}\}$ such that $\Phi$ evaluates to \texttt{true}.
\begin{theorem}
\label{thm:hardness}
It is NP-hard to decide whether $\station(k,P,Q)\le \eps$ for given $k>0$ and $\eps>0$, even when $P$ and $Q$ are weakly simple polygonal chains.
\end{theorem}
\begin{proof}
\ShoLong{
We present here only an overview of the proof. The details can be found in the Appendix.
}{
We start with a quick overview of the reduction and then continue with the details.
}
Given an instance $A$ of \textsc{Planar-Rectilinear-3SAT}, we build an instance $B$ of our problem producing two polygonal chains, $P$ and $Q$, as shown in Figure~\ref{fig:reduction}. The chain $P$ ($Q$) is represented by a blue (red) curve. Black edges represent overlap between $P$ and $Q$. We set $\eps:=1$, and design $P$ and $Q$ so that the length of almost every edge is an integer. That allows us to compute locally optimal solutions along the black edges that require a consistent choice of station placement alternating between blue and red stations, which in turn establishes a lower bound on the number of stations. 
We set the parameter $k$ so that every solution must meet that lower bound.
In the variable gadget, a concatenation of \emph{literal gadgets} (Figure~\ref{fig:gadgets} (a)) must alternate consistently in order to achieve this lower bound. 
The choice of whether to start with a blue or a red station encodes the truth value of the variable. The \emph{separation gadget} (Figure~\ref{fig:gadgets} (c)) allows choosing truth values for each variable independently. In the \emph{clause gadget} (Figure~\ref{fig:gadgets} (d)), a subchain of $Q$ (near $p_5$) can be covered by a blue station of the alternation of a literal gadget if the literal evaluates to \texttt{true}.
If all literals in the clause evaluate \texttt{false}, then either an additional station is needed or $\eps$ has to be increased. Hence, $\station(k,P,Q)\le \eps$ if and only if the instance $A$ admits a positive solution.
More precisely, we construct our curves and prove correctness of the reduction as follows:
	\begin{figure*}[ht]
		\centering
		\includegraphics[width=\ShoLong{.7}{}\textwidth]{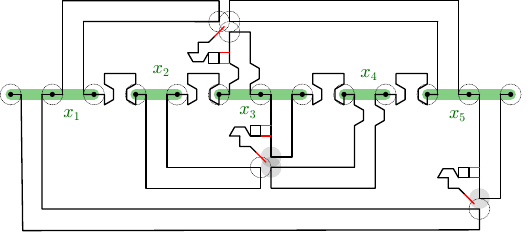}
		\caption{Reduction from the instance $(x_1\vee x_3\vee x_5)\wedge(x_1\vee \overline{x_5})\wedge(x_2\vee\overline{x_3}\vee\overline{x_4})$.
		The segments in the $x$-axis corresponding to the five variables are shown in green.}\label{fig:reduction}
	\end{figure*}
	\begin{figure}[ht]
		\centering
		\includegraphics[width=\ShoLong{}{0.6}\columnwidth]{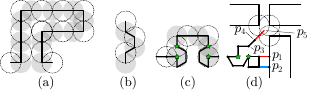}
		\caption{(a) Literal, (b) negation, (c) separation, and (d) clause gadgets.
		}\label{fig:gadgets}
	\end{figure}
%

\paragraph{Construction.}
Let $A$ be an instance of \textsc{Planar-Rectilinear-3SAT}. Without loss of generality, we may assume the following properties about the rectilinear drawing of instance $A$: The drawing lies on an integer grid. Each variable $x_i$ is represented by a line segment of length $\deg(x_i)$ on the $x$-axis. The variable segments are one unit apart. For clauses containing $3$ literals, the corresponding vertex is located vertically above/below the segment of its middle variable, and is incident to two edges with exactly one bend and one straight vertical edge.
For clauses containing $2$ literals, the corresponding vertex is vertically above/below its rightmost variable and is incident to one edge with one bend and one straight edge.
Next, scale up the given rectilinear drawing of instance $A$ vertically by a factor of $8$ and horizontally by $4$.
Then replace each edge between a variable and a clause by a literal gadget (Figure~\ref{fig:gadgets} (a)) that starts and ends with unit horizontal segment along the $x$-axis and remains in the unit-neighborhood (in $L_1$ norm) of the corresponding edge in $A$.

If the edge is a straight-line segment (has one bend) and corresponds to a positive (negative) literal, then replace the subchain of length $4$ with endpoints on the $x$-axis by the \emph{negation gadget} (Figure~\ref{fig:gadgets} (b)).
The gadget is made of $5$ unit-length line segments, $3$ of which are vertical, and the remaining two segments have slopes $\frac{\sqrt{3}}{3}$ and $-\frac{\sqrt{3}}{3}$, resp.,
so that the height of the gadget is $4$. 

We add the separation gadget (Figure~\ref{fig:gadgets} (c)) between every pair of consecutive literal gadgets that correspond to different variables. 
The gadget has width $3$ (starting and ending at vertices marked with a green star), so we add one horizontal unit segment to the left literal in order to connect the gadgets.
In both $P$ and $Q$, the gadget contains $6$ vertical unit segments, $4$ unit segments at slopes $\pm\frac{\sqrt{3}}{3}$, and one horizontal segment of length 3.
	
For each vertical edge in the rectilinear drawing of $A$, we add a clause gadget (Figure~\ref{fig:gadgets} (d)) at the vicinity of the clause vertex as follows. Assume that the clause is drawn in the upper half-plane, reflecting the construction through the $x$-axis otherwise. 
Let $p_1$ be 2 units below the left corner of the literal gadget corresponding to the vertical straight edge incident to the clause in $A$. 
The path connecting the two green stars  traces three consecutive edges of a regular hexagon of unit-length sides.
Set $p_2=p_1+\left(0,-1\right)$, 
$p_3=p_1+\left(-2,1\right)$,
$p_4=p_1+\left(-\sqrt{2},3-\sqrt{2}\right)$, and 
$p_5=p_1+\left(-\frac{1}{2},3-\frac{1}{2}\right)$.
The remaining points lie in the integer grid 
and can be easily recovered from Figure~\ref{fig:gadgets} (d).
The clause gadget consists of a subchain of $P$ consisting of two paths between $p_1$ and $p_4$, and a subchain of $Q$ consisting of two paths between $p_1$ and $p_5$, as shown in Figures~\ref{fig:local-solution}~(c) and (d).
We split the chains at the green star closest to $p_1$ and assign the parts adjacent to the literal gadget to that gadget, i.e., we consider that the clause gadget starts at the green star.
Finally, we subdivide the literal gadget at $p_1$ into two subchains of $P$ and $Q$ each.

After combining the subchains of the gadgets, described above, we obtain two weakly simple polygonal chains, $P$ and $Q$. 
We call the endpoints of $P$ and $Q$, and the points marked by a green star in Figure~\ref{fig:gadgets} \emph{critical points}.
An orthogonal path between critical points is called \emph{critical path}.
Let $\ell_c$ be the length of a critical path $c$, and let $C_i$ be the set of critical paths in the literal gadgets corresponding to the $i$-th variable $x_i$. Set $k=\sum_{i=1}^n \sum_{c\in C_i}(\frac{\ell_c+1}{2})+ 10(n-1)+12m$, and $\eps=1$.
This concludes the construction of instance $B$.

\paragraph{Correctness.}
We can describe a solution for $B$ as follows.
Subdivide $P$ (resp., $Q$) into $k_1$ ($k_2$) subchains $P_1,\ldots ,P_{k_1}$ ($Q_1,\ldots ,Q_{k_2}$) called \emph{pitches} and let $p_i$ and $p_{i+1}$ ($q_i$ and $q_{i+1}$) be the endpoints of $P_i$ ($Q_i$), called \emph{stations}. 
The pitches form a solution if $k_1+k_2\le k$ and for every pitch $P_i$ ($Q_i$) lies in a disk of radius $\eps=1$ centered at a  station in $Q$ ($P$). In Figs.~\ref{fig:reduction}, \ref{fig:gadgets}, and \ref{fig:local-solution} the centers of circular disks represent stations. A blue (red) disk is centered at some station $p_i$ ($q_i$).

	\begin{figure*}[h]
		\centering
		\includegraphics[width=\ShoLong{.85}{}\textwidth]{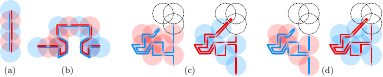}
		\caption{Locally optimal solutions.}\label{fig:local-solution}
	\end{figure*}

($\Rightarrow$) First, assume that $A$ is a positive instance.
For each variable $x_i$ assigned \texttt{true}, subdivide the subchains of $P$ and $Q$ that correspond to the literal gadgets of $x_i$ by placing blue (red) stations at the center of dashed (grey) circles as shown in Figure~\ref{fig:gadgets} (a). 
For \texttt{false}-valued variables, switch red and blue in the previous sentence.
Since the literal gadgets have even length in $P$ and $Q$ by construction, the alternation of blue and red stations along a variable is consistent and the literal gadget is subdivided into $\frac{\ell_i}{2}+1$ pitches (Figure~\ref{fig:local-solution}~(a)). For every separation gadget, add a blue (red) station at the center of each dashed (grey) disk shown in Figure~\ref{fig:gadgets}~(c). 
Notice that both blue and red stations are placed on critical points.
In our instance, this subdivides the separation gadget into a total of 10 pitches (in $P$ and $Q$ combined), shown in Figure~\ref{fig:local-solution} (b) (the figure shows 5 additional pitches that are counted as part of the adjacent literal gadgets).
Finally, for each clause gadget, subdivide $P$ and $Q$ as shown in Figure~\ref{fig:local-solution} (c) or (d) if the vertical literal (that corresponds to the vertical edge incident to the clause in $A$) evaluates to \texttt{true} or \texttt{false}, respectively. The number of pitches created that are not counted in the literal gadgets is 12. By construction, every pitch in the subdivision is within distance 1 from a station of the opposite color. Therefore, instance $B$ admits a solution.

($\Leftarrow$) Now assume that $B$ admits a solution $\mathcal{S}$.
We show that, as $k$ is the sum of local lower bounds on the number of pitches in a solution, $\mathcal{S}$ locally resembles Figure~\ref{fig:local-solution}. We fix a direction of traversal of $P$ and $Q$ from its left to its right endpoint.

\begin{lemma}
\label{lem:criticalPoints}
If $B$ admits a solution, there exists a solution $\mathcal{S}$ in which there is a blue and a red station at every critical point.
\end{lemma}
\begin{proof}
Endpoints of $P$ and $Q$ are endpoints of pitches in $\mathcal{S}$ and are therefore stations.
For the remaining critical points, we argue, without loss of generality, for points $p_1$ and $p_5$ in the separation gadget.
Point $p_2$ in $P$ ($Q$) must be covered by a red (blue) station $r$ ($b$) in the path $(p_2,p_3,p_4)$.
Let $Q_i$ ($P_j$) be the pitch starting at $r$ ($b$).
Its other endpoint must precede $p_5$ as $r$ ($b$) must be covered by a blue station on the path from $p_1$ to $p_4$.
Then, we can move both $r$ and $b$ to $p_1$ without affecting the solution because $Q_i$ and $P_j$ remain covered (by $b$ and $r$, resp.), and the pitches to the left become shorter and therefore are still covered. Similarly, the other endpoint of both $Q_i$ and $P_j$ can be moved to $p_5$.
\end{proof}

\begin{lemma}
\label{lem:literalLowerBound}
A pair of critical paths $P'$ and $Q'$ (i.e., orthogonal subpaths of $P$ and $Q$, respectively, between two critical points) of lengths $\ell_P$ and $\ell_Q$ require at least $\frac{\ell_P+\ell_Q}{2}+1$ pitches in $\mathcal{S}$. 
If this lower bound is attained, the stations are laid out as in Figure~\ref{fig:local-solution} (a), alternating between red and blue stations one unit apart along the path. 
\end{lemma}
\begin{proof}
The endpoints of $P'$ ($Q'$) are blue (red) stations by Lemma~\ref{lem:criticalPoints}.
By construction, $P'$ ($Q'$) are orthogonal paths whose vertices lie on the integer grid. 
We first argue that the minimum number of pitches of $P'$ in a solution is $\ell_P/2$ and assume that $\ell_P$ is even.
Let $\mathcal{S}_P$ be a minimum cardinality partition of $P'$ such that every subchain in $\mathcal{S}_P$ is within unit distance from some point in $Q$.
We claim that all blue stations in $\mathcal{S}_P$ lie on the integer grid, and every pitch has length $2$.
We prove the claim by induction on the length of $P'$.
Assume $v$ is the first blue station in $\mathcal{S}_P$ not in the integer grid or that the subchain $P_i$ in $\mathcal{S}_P$ have length different than two. 
Let $v'$ ($v''$) be its successor (predecessor) blue stations.
Since $v'$ lies on the integer grid by the induction hypothesis, a unit disk centered at a point in $Q$ that contains $v'$ can cover a subchain of $P'$ of length at most $2$ by construction.
The length of such a path is exactly $2$ when $v'$ is on the boundary of the disk.
We expand $P_i$ by moving $v$ along $P'$.
Hence, $|\mathcal{S}_P|=\ell_P/2$.

In order to achieve the lower bound of $\frac{\ell_P+\ell_Q}{2}+1$, either $P'$ or $Q'$ must be partitioned optimally.
Without loss of generality, let $P'$ be  optimally partitioned.
By the previous claim, the subchains of $P'$ have length $2$.
By construction, the at least one subchain of $Q'$ must have length $1$ so that all subchains of $P'$ are within unit distance from a red station.
Therefore the number of subchains of $Q'$ is larger than $\ell_Q/2$.
If the lower bound is achieved, the stations alternate as claimed.
A similar argument proves the claim for odd $\ell_P$ and $\ell_Q$ with the exception that both $P'$ and $Q'$ will each contain a pitch of length $1$ and the remaining pitches will have length $2$.
\end{proof}

We now establish lower bounds for the separation and clause gadgets.
Assume that $\mathcal{S}$ satisfies Lemma~\ref{lem:criticalPoints}.
A direct consequence of Lemmas~\ref{lem:criticalPoints} and \ref{lem:literalLowerBound} is that the number of pitches in the separation gadget is at least $10$.
The lower bound can be achieved as shown in Figure~\ref{fig:local-solution}~(b).
In the clause gadget, the pitch of $P$ with an endpoint at the critical point must have the other endpoint before $p_3$.
Notice that there is a neighborhood of $p_3$ in $Q$ that can only be covered by a blue station on interior of the path $(p_3,p_4,p_3)$ of $P$.
Hence, $\mathcal{S}$ contain at least 3 pitches of $P$ between critical points.
The nighborhood of $p_5$ in $Q$ can only be covered by a blue station on a literal gadget because $P$ turns around at $p_4$ and $\|p_4p_5\|>1$.
Then, we can show in a similar way that 
$\mathcal{S}$ contain at least 5 pitches of $Q$ between critical points.
Then, counting the pitches in the half-hexagon, the clause gadget requires at least 12 pitches.
Such bound can be achieved as shown in Figure~\ref{fig:local-solution}~(c)--(d) if an adjacent literal has a blue station represented by one of the dashed circles.

Since $k$ is the sum of all local lower bounds, every integer-length path between critical points is partitioned as in Lemma~\ref{lem:literalLowerBound}.
We now show that the alternation in the literal gadgets must be consistent at the vicinity of a clause gagdet.
Refer to Figure~\ref{fig:local-solution}~(c) and (d).
Let $P'$ and $P''$ ($Q'$ and $Q''$) be the upper and lower critical paths of $P$ ($Q)$ in the literal gadget adjacent to the right critical point.
By construction they all have even lengths.
If $P'$ is is optimally partitioned, by Lemma~\ref{lem:literalLowerBound} the first and last pitches of $Q'$ have length $1$ (as in Figure~\ref{fig:local-solution}~(c)).
For contradiction, assume that $P''$ is not partitioned optimally, and, therefore, the corresponding pitches of $P''$ and $Q''$ are as in Figure~\ref{fig:local-solution}~(d).
Then, the second pitch of $Q''$ does not lie within a unit distance from a blue station.
Similar arguments show a contradiction for the other cases in which the alternation is not consistent. 
Therefore, each variable has a well-defined truth value associated with the station alternation.
%
Additionally, every clause is adjacent to a literal that evaluates to \texttt{true}.
Hence, converting the alternation into a truth assignment for the variables results in a solution for $A$.
\end{proof}
\section{Approximation Algorithms}
\label{sec:apx}

In this section, we show that for two polygonal chains, $P$ and $Q$, and a threshold $\eps>0$, 
we can approximate the minimum $k\in \mathbb{N}$ for which $\station(k,P,Q)\leq \eps$ up to a factor of 2. 
Recall that $\station(k,P,Q)\leq \eps$ if and only if there exist a set $\mathcal{S}$ of $k$ axis-parallel line segments in the free space $F_\eps(P,Q)$ such that 
$\proj_x(\bigcup \mathcal{S})=\proj_x(F_\eps(P,Q))$,
$\proj_y(\bigcup \mathcal{S})=\proj_y(F_\eps(P,Q))$,
and the projections of the segments onto the two 
coordinate axes have pairwise disjoint relative interiors.

The condition that the projections of segments in $\mathcal{S}$ are interior-disjoint is crucial. Without this condition, the problem would be separable, and we could find an optimal solution efficiently: Let $\opt_x$ be a minimum cardinality set of horizontal segments in $F_\eps(P,Q)$ such that $\proj_x(\bigcup \opt_x)=\proj_x(F_\eps(P,Q))$, and $\opt_y$ a minimum set of vertical segments in $F_\eps(P,Q)$ such that $\proj_y(\bigcup \opt_y)=\proj_y(F_\eps(P,Q))$.

\begin{obs}\label{obs:opt}
The set $\mathcal{S}=\opt_x\cup \opt_y$ is a minimum set of axis-parallel segments 
such that $\proj_x(\bigcup \mathcal{S})=\proj_x(F_\eps(P,Q))$, and $\proj_y(\bigcup \mathcal{S})=\proj_y(F_\eps(P,Q))$.
\end{obs}

\begin{proof}
Suppose $\opt_x\cup \opt_y$ is not minimal, i.e., there exists a smaller such set $\mathcal{S}'$ of axis-parallel segments whose $x$- and $y$-projection equals that of $F_\eps(P,Q)$. Partition $\mathcal{S}'$ into subsets of horizontal and vertical segments, say $\mathcal{S}'_x$ and $\mathcal{S}'_y$. Then $|\mathcal{S}|'<|\opt_x|+|\opt_y|$ implies $|\mathcal{S}'_x|<|\opt_x|$ or $|\mathcal{S}'_y|<|\opt_y|$, contradicting the minimality of $\opt_x$ or $\opt_y$.
\end{proof}

Given a set of axis-parallel line segments, we can eliminate intersections between the relative interiors of their $x$- and $y$-projections at the expense of increasing the number of segments by a factor of at most 2.

\begin{lemma}\label{lem:approx}
There exists a set $\mathcal{S}$ of at most $2(|\opt_x|+|\opt_y|)$
axis-parallel segments in $F_\eps(P,Q)$ such that
$\proj_x(\bigcup \mathcal{S})=\proj_x(F_\eps(P,Q))$,
$\proj_y(\bigcup \mathcal{S})=\proj_y(F_\eps(P,Q))$,
and the projections of the segments onto the two 
coordinate axis have pairwise disjoint relative interiors.
\end{lemma}
\begin{proof}
We may assume, by truncating the segments in $\opt_x$ and $\opt_y$, if necessary, that the $x$-projections of segments in $\opt_x$ are interior-disjoint, and the $y$-projections of segments in $\opt_y$ are also interior-disjoint. Then the supporting line of each horizontal segment in $\opt_x$ intersects the interior of at most one vertical segment in $\opt_y$, and vice versa.
Consequently, the supporting lines of the segments in $\opt_x$ (resp., $\opt_y$) 
jointly subdivide the segments in $\opt_y$ (resp., $\opt_x$) into at most 
$|\opt_x|+|\opt_y|$ pieces. The total number of resulting axis-parallel segments is $2(|\opt_x|+|\opt_y|)$, as required.
\end{proof}

It remains to show how to compute $\opt_x$ and $\opt_y$ efficiently. 
We first observe that a greedy strategy finds $\opt_x$ (resp., $\opt_y$) from a set of maximal horizontal (resp., vertical) segments in $F_\eps(P,Q)$.

\paragraph{A Greedy Strategy.} 
Input: A set $\mathcal{H}$ of horizontal line segments in $\mathbb{R}^2$. 
Output: a subset $\mathcal{S}\subset \mathcal{H}$ such that $\proj_x(\bigcup \mathcal{S})=\proj_x(\bigcup \mathcal{H})$. Initialize $\mathcal{S}:=\emptyset$;
and let $L$ be a vertical line through the leftmost points in $\bigcup \mathcal{H}$.
Let $L^-$ be the closed halfplane on the left of $L$. 
While $\proj_x(\bigcup \mathcal{S})\neq \proj_x(\bigcup \mathcal{H})$, do:
Let $s\in \mathcal{H}$ be a segment whose left endpoint is in $L^-$ and whose right endpoint has maximal $x$-coordinate. Put $\mathcal{S}\leftarrow \mathcal{S}\cup \{s\}$; let $L\leftarrow$the vertical line through the right endpoint of $s$, and $\mathcal{H}\leftarrow \{h\in \mathcal{H}: h\not\subset L^-\}$.

\begin{obs}
Given a set $\mathcal{H}$ of horizontal segments, the above greedy algorithm returns a minimum subset $\mathcal{S}\subset \mathcal{H}$ such that $\proj_x(\bigcup \mathcal{S})=\proj_x(\bigcup \mathcal{H})$.
\end{obs}
\begin{proof}
At each iteration of the while loop, we maintain the following invariant: 
$\mathcal{S}$ is a minimal subset of $\mathcal{H}$ such that 
$\proj_x(\bigcup \mathcal{S})=\proj_x((\bigcup \mathcal{H})\cap L^-)$.
\end{proof}

The implementation of the above greedy algorithm is straightforward when $\mathcal{H}$ is finite. However, the set $\mathcal{H}$ of maximal horizontal segments in the free space $F_\eps(P,Q)$ may be infinite.

\begin{lemma}\label{lem:optxy}
Let $P$ and $Q$ be polygonal chains with $m$ and $n$ segments, respectively, and let $\eps>0$.
Then a set $\opt_x$ can be computed in output-sensitive $O((|\opt_x|+m)n)$ time.
\end{lemma}

\begin{proof}
Let $\mathcal{H}$ be the set of maximal horizontal segments in the free space $F_\eps(P,Q)$. 
To implement the greedy algorithm above, we describe a data structure that supports the following query: Given a vertical line $L$, find a segment $s\in \mathcal{H}$ whose left endpoint is in $L^-$ and whose right endpoint has maximal $x$-coordinate. 

Recall from Section~\ref{sec:feasibility} that the parameter space $[0,1]^2$ is subdivided into $mn$ axis-parallel cells $C_{i,j}$. In each cell, $C_{i,j}\cap F_\eps(P,Q)= C_{i,j}\cap E_{i,j}$, where $E_{i,j}$ is either an ellipse or a slab parallel to the diagonal of $C_{i,j}$. 

Let a vertical line $L$ be given, and assume that it intersects the cells $C_{i,j}$, for $j=1,\ldots ,n$. In each of these $n$ cells, compute the intersections $\ell_{i,j}=L\cap C_{i,j}\cap E_{i,j}$, and the set $R_{i,j}$ of points in $C_{i+1,j}\cap E_{i,j}$ that can be connected to $\ell_{i,j}$ by a horizontal line segment within $C_{i,j}\cap E_{i,j}$. 
If none of the sets $R_{i,j}$ touches the right edge of the cell $C_{i,j}$, then 
take a rightmost point $r$ in $\bigcup_{j=1}^n R_{i,j}$, and report a maximal horizontal line segment in $F_\eps(P,Q)$ whose right endpoint is $r$; this takes $O(n)$ time.
Otherwise consider the vertical line $L'$ passing through the right edges of the cells $C_{i,j}$ ($j=1,\ldots ,n$); and let $\ell'_{i,j}=L'\cap R_{i,j}$. 
We can repeat the above process in cells $C_{i+1,j}$ ($j=1,\ldots ,n$) with lines $\ell_{i,j}'$
in place of $\ell_{i,j}$.  Ultimately, we find a rightmost point $r\in F_\eps(P,Q)$ that 
can be connected to a point in $L$ within $F_\eps(P,Q)$. 

Each query $L$ takes $O(n)$ time if it finds $r$ within a cell $C_{i,j}$ stabbed by $L$; and 
$O(nt)$ time if it finds $r$ in a cell $C_{i+t,j}$ for some $t=1,\ldots ,m-i$. 
Since the $x$-coordinates of the query lines are strictly increasing, 
the total running time for $|\opt_x|$ queries is $O((|\opt_x|+m)n)$ time, as claimed.
\end{proof}

\begin{theorem}
Let $P$ and $Q$ be polygonal chains with $m$ and $n$ segments, respectively, and let $\eps>0$.
Then we can approximate the minimum $k$ such that $\station(k,P,Q)\leq \eps$ within a factor of 2 in output-sensitive $O(k(m+n)+mn)$ time.
\end{theorem}
\begin{proof}
Compute the free space $F_\eps(P,Q)$ in $O(mn)$ time. If $\proj_x(F_\eps(P,Q))\neq [0,1]$ or $\proj_y(F_\eps(P,Q))\neq [0,1]$, then report that $\station(k,P,Q)> \eps$ for every $k\in \mathbb{N}$. Otherwise, compute $\opt_x$ and $\opt_y$ by Lemma~\ref{lem:optxy}
in $O(mn+|\opt_x|n+m|\opt_y|)$ time. 
We have $k\leq |\opt_x|+|\opt_y|$ by Observation~\ref{obs:opt}.
Lemma~\ref{lem:approx} yields a set $\mathcal{S}$ of at most $2(|\opt{S}_x|+|\opt{S}_y|)$
axis-parallel segments in $F_\eps(P,Q)$ such that
$\proj_x(\bigcup \mathcal{S})=\proj_x(F_\eps(P,Q))$,
$\proj_y(\bigcup \mathcal{S})=\proj_y(F_\eps(P,Q))$,
and the projections of the segments onto the two 
coordinate axes have pairwise disjoint relative interiors.
In particular, $\station(|\mathcal{S}|,P,Q)\leq \eps$,
and so $k\leq |\mathcal{S}|\leq 2k$, as required. 
The running time of our algorithm is 
 $O(mn+|\opt_x|n+m|\opt_y|)\subset O(mn+k(m+n))$.
\end{proof} 

\section{Conclusion}
\label{sec:con}

We have introduced the rock climber distance $\rock(k,P,Q)$ and the $k$-station distance $\station(k,P,Q)$ between two polygonal chains in the plane. The rock climber distance combines properties of the continuous and discrete Fr\'echet distance: It corresponds to a coordinated motion of two agents traversing the two chains where only one agent moves at a time.
Our results raise several open problems, we present some of them here. 
\begin{itemize}
    \item Can we efficiently approximate $\station(k,P,Q)$ for a given $k$
    and given polygonal chains $P$ and $Q$?
    \item In Section~\ref{sec:apx}, we described a 2-approximation algorithm for finding the minimum $k$ for which $\station(k,P,Q)\leq \eps$. Can the approximation ratio be improved? Does the problem admit a PTAS?
    	\begin{figure}[ht]
		\centering
		\includegraphics[scale=0.8]{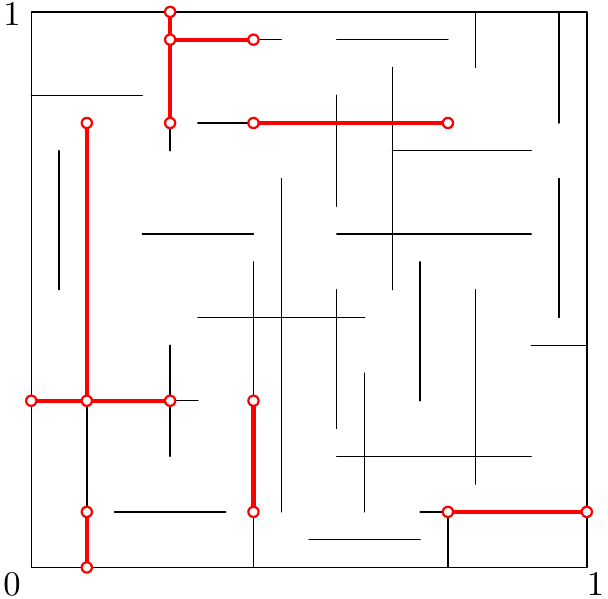}
		\caption{An instance of the compatible axis-parallel segment cover problem. A solution of size 10, 
		is shown in red (bold).}\label{fig:pseudo}
	\end{figure}
    \item A discretization of the previous problem leads to the \emph{compatible axis-parallel segment cover problem}: Instead of the free~space $F_\eps(P,Q)$, we are given a set $F\subset [0,1]^2$ as a union of $n$  axis-aligned line segments, and ask for the minimum $k\in \mathbb{N}$ such that $F$ contains $k$ axis-parallel line segments whose vertical and horizontal projections, respectively, have pairwise disjoint relative interiors, and jointly cover the unit interval $[0,1]$. See Fig.~\ref{fig:pseudo}. 
    The conditions on disjoint relative interiors is crucial,
    and can be formulated as a geometric set cover problem  with conflicts~\cite{GeomSetCover-Conflicts-17}, or with unique coverage~\cite{UniqueGeomCover-15,AshokRG17}.
     Our NP-hardness and 2-approximation results extend to this problem. 
     Can the approximation ratio be improved? 
     Is the problem APX-hard?

\end{itemize}

\bibliography{bibliography}

\end{document}